\documentclass[11pt,final,onecolumn]{ieeetran}
\usepackage{color}
\usepackage{graphicx}
\usepackage{amsmath}
\usepackage{authblk}
\usepackage{times}
\usepackage{subfigure}
\usepackage{cite}
\usepackage{amssymb}
\usepackage{geometry}
\usepackage{psfrag}
\renewcommand{\baselinestretch}{1.7}

\def\Prob{{\rm P}\,}

\def\Exp{{\mathbb{E}}\,}
\def\tr{{\rm tr}\,}

\def\be{\begin{equation}}
\def\ee{\end{equation}}
\def\ba{\left[\begin{array}}
\def\ea{\end{array}\right]}
\def\bea{\begin{eqnarray}}
\def\eea{\end{eqnarray}}

\newcommand{\mb}[1]{\mathbf{#1}}
\newcommand{\mr}[1]{\mathrm{#1}}
\newcommand{\mc}[1]{\mathcal{#1}}
\newcommand{\ol}[1]{\overline{#1}}

\def\ba{{\bf a}}

\newtheorem{theorem}{\textbf{Theorem}}

\newtheorem{corollary}{\textbf{Corollary}}
\newtheorem{lemma}{\textbf{Lemma}}

\begin{document}
\title{Short-term Performance Limits of MIMO Systems with Side Information at
the Transmitter
}
%
\date{}
\author{Liangbin Li, Hamid Jafarkhani}
\affil{Center for Pervasive Communications \& Computing, University
of California, Irvine\vspace{-0.3in}} \maketitle

\begin{abstract}
The fundamental performance limits of space-time block code (STBC)
designs when perfect channel information is available at the
transmitter (CSIT) are studied in this report. With CSIT, the
transmitter can perform various techniques such as rate adaption,
power allocation, or beamforming. Previously, the exploration of
these fundamental results assumed long-term constraints, for
example, channel codes can have infinite decoding delay, and power
or rate is normalized over infinite channel-uses. With long-term
constraints, the transmitter can operate at the rate lower than the
instantaneous mutual information and error-free transmission can be
supported. In this report, we focus on the performance limits of
short-term behavior for STBC systems. We assume that the system has
\emph{block power constraint}, \emph{block rate constraint}, and
\emph{finite decoding delay}. With these constraints, although the
transmitter can perform rate adaption, power control, or
beamforming, we show that decoding-error is unavoidable. In the high
SNR regime, the diversity gain is upperbounded by the product of the
number of transmit antennas, receive antennas, and independent
fading block channels that messages spread over. In other words,
\emph{fading cannot be completely combatted with short-term
constraints.} The proof is based on a sphere-packing argument.

\end{abstract}

\renewcommand{\baselinestretch}{1}

\section{Introduction}
There are much interest in the research of side information at the
transmitter for multi-input multi-output (MIMO) communication
systems in fading channels. Various techniques have been proposed to
improve system performance using channel information at the
transmitter (CSIT): e.g., beamforming and precoder designs, power
allocation and rate adaption methods. The fundemental performance
limits of these techniques have been studied extensively. In terms
of decoding error probability, the performance limits can be
catergorized as:
\begin{enumerate}
  \item \emph{Full spatial-and-temproal diversity:} The receiver observes decoding errors. The error decays as a polynomial function of the signal-to-noise ratio (SNR) in the high SNR
region. The exponent is defined as the \emph{diversity gain}, which
is limited by the product of the number of transmit antennas,
receive antennas, and independent fading block channels.
  \item \emph{Infinite-diversity:} The receiver observes decoding errors.
The error decays exponentially with SNR, like the decoding error in
an additive white Gaussian noise (AWGN) channel. The diversity gain
is equal to infinite.
  \item \emph{Error-free transmission}: There is no decoding error at the receiver.
\end{enumerate}
Different system constraints result in different performance limits.
One of the constraints that will affect the performance limit is the
\emph{decoding delay}, which is defined as the number of
channel-uses that the receiver can wait before decoding messages. In
other words, it is the number of channel-uses where each information
bit spreads. Regarding to decoding delay, there are long-term power
and rate constraints as well as short-term power and rate
constraints.

Long-term constraints are determined by averaging over all the
channel states. A system employing rate adaption with long-term rate
constraint, fixed power, and finite decoding delay, can achieve
infinite diversity in fading channels\cite{Ca72}. The long-term rate
constraint can be written as
\begin{align*}
\int R(\mb{h})f(\mb{h})d\mb{h}\le R,
\end{align*}
where the instantaneous transmission rate $R(\mb{h})$ depends on the
equivalent channel vector $\mb{h}$, and $f(\mb{h})$ and $R$ denote
the probability density function (PDF) of $\mb{h}$ and average
transmission rate, respectively. When power allocation with
long-term constraint is used for a system with infinite decoding
delay, error-free transmission can be achieved for fixed rate. The
notion of \emph{delay-limited capacity} is defined as the
transmission rate that the system can reliablly support for all
channel realizations under long-term power allocation at the
transmitter\cite{Cai99}. The long-term power constraint can be
written as
\begin{align*}
\int P(\mb{h})f(\mb{h})d\mb{h}\le P,
\end{align*}
where the transmit power $P(\mb{h})$ is a function of the equivalent
channel vector $\mb{h}$, and $P$ denotes the average power. It is
shown that for single-input single-output (SISO) systems, the
delay-limited capacity is zero, while for MIMO systems, there is a
nonzero delay-limited capacity.

The long-term constraints are impractical assumptions for real-world
implementation. Some applications, for example, video transmission,
are delay sensitive and require finite decoding delays. In addition,
infinite-length codewords can lead to extremely high decoding
complexity. A power allocation with the long-term constraint may
result in an infinite peak power, which is intolerable for
electronic devices. Rate adaption with the long-term constraint
usually cuts off systems when the receive SNR is low. Nevertheless,
delay-sensitive applications need a minimum rate even though
channels are in deep fading. For these reasons, power or rate
constraint needs to be determined by averaging over a finite number
of channel-uses. These constraints define the short-term behavior of
communication systems. A mixture of both short-term and long-term
constraints have been considered in the literature. For example, a
system using long-term power allocation, fixed rate, and finite
decoding delay achieves infinite diversity in fading
channels\cite{Sha08}. We summarize the diversity performance with
respect to different types of constraints and a finite decoding
delay in Table \ref{table-known}.
\begin{table}[h]
  \centering
  \caption{Diversity performance for MIMO systems with finite decoding delay}\label{table-known}
\begin{tabular}{|c|c|c|}
  \hline
  Diversity gain & Long-term power constraint & Short-term power constraint
\\ \hline
  Long-term rate constraint & infinite & infinite\cite{Ca72} \\
  Short-term rate constraint & infinite\cite{Sha08} & {{unknown}} \\
  \hline
\end{tabular}
\end{table}

Note that short-term constraints are special cases of long-term
constraints. The performance achievable under short-term constraints
is achievable under long-term constraints as well. In Table
\ref{table-known}, the results under both long-term power and rate
constraints are straightforward extensions from the results in
\cite{Ca72, Sha08}. Conversely, the results under long-term
constraints cannot be used for a system with short-term constraints.
To the best of our knowledge, there is no rigorous results under
both short-term power and rate constraints.

This report aims at exploring the performance limits of transmitter
controls, for example, beamforming, power allocation, rate adaption
in conjunction with space-time block code (STBC) designs, under the
assumptions of perfect CSIT, finite decoding delay, and short-term
power and rate constraints. To allow power allocation and rate
adaption within the scope of short-term constraints, we introduce
\emph{transmission delay}, the number of channel-uses where power
and rate are constrained. The transmission delay is the sum of all
decoding delays for STBCs. Assume that a system needs to transmit
$RT$ bits of information using $L$ block codes, each with a decoding
delay constraint of $D_l$ $(l=1,\ldots, L)$, where
$T=\underset{l=1:L}{\sum} D_l$, denotes the transmission delay, and
$R$ denotes the average rate constraint. The transmission for $T$
channel-uses has a \emph{block power constraint} of $PT$, where $P$
denotes the average power constraint. \emph{We assume that the
channel information of all $T$ channel-uses is given noncausually to
the transmitter.} Then, rate adaption or power allocation can be
conducted within the scope of the transmission delay. For example,
the transmitter is aware that the Frobenius norm of the channel
matrix in the first $T/2$ channel-uses is higher than that in the
second $T/2$ channel-uses. Then, a rate adaption and power
allocation scheme can send all $RT$ information bits in the first
$T/2$ channel-uses using a power equal to $2P$. In the second $T/2$
channel-uses, the transmitter keeps silent.

Note that any transmitter control scheme with short-term constraints
on $T$ channel-uses can be viewed as a realization of concatenated
STBC with fixed-rate constraint $R$ bits/channel-use, sum-power
constraint $PT$, and decoding delay constraint $T$ channel-uses. The
performance limits of the concatenated block code can be applied
directly to any transmitter control scheme with short-term
constraints. Therefore, we study the performance limits of
fixed-rate STBC designs with sum-power and decoding delay
constraints when CSIT is available.

For fixed-rate and finite decoding-delay designs, there are two
approaches in AWGN channels. Gallager uses the \emph{random coding
argument} to show that the decoding error probability drops
exponentially with the code length\cite{Gallager}. The random codes
show an achievable performance of error probability, which implies
the existence of a good code outperforming it. Thus, an upperbound
on the achievable error probability is provided by the random codes.
Extension of this approach to MIMO systems can be found in
\cite{Ko05}. The other approach uses \emph{sphere-packing argument}.
It models code-design as a sphere-packing problem. The converse of
Shannon capacity can be shown using this argument\cite{Wozencraft}.
It can be extended to MIMO systems with no CSIT to provide a
lowerbound on the codeword decoding error probability for any STBC
design\cite{Fo03}.

To find the fundemantal performance limits of MIMO systems with
CSIT, we take the sphere-packing approach to calculate a lowerbound
on the codeword decoding error probability given short-term
constraints. The main contribution of this report can be summerized
as follows:
\begin{enumerate}
  \item For $M\times N$ MIMO systems where each message is encoded over $K$ independent fading blocks, we
show that the maximum diversity is $MNK$ under short-term
constraints even though the transmitter has perfect CSIT. The result
is the same as the scenario when channel information is not
available at the transmitter.
  \item Although \cite{Fo03} assumes a finite decoding delay, the authors assume sphere-hardening at the receiver\cite{Wozencraft}, which implies that the channel codes have infinite-length. Thus, the proposed bound in \cite{Fo03} is not a strict lowerbound on decoding error probability. In this report, we
avoid using sphere-hardening arguments. Thus, the results in this
report rely completely on finite decoding delays.
\end{enumerate}

We do not claim that our lowerbound is achievable. However, the
resulting diversity can be achieved using STBCs. The negative result
shows what the system cannot achieve, in other words, the
performance limits of STBC design.

The rest of this report is organized as follows. In Section
\ref{sec-model}, we explain the system model. Section
\ref{sec-bound} provides a lower bound on error probability. In
Section \ref{sec-numeric}, we show the simulated performance of the
derived lowerbounds. Conclusions are provided in Section
\ref{sec-conclusion}.

\textit{Notation}: For a matrix $\mb{A}$, let $\mb{A}^*$,
$\tr\left(\mb{A}\right)$, and $\|\mb{A}\|$ denote its Hermitian,
trace, and Frobenius norm, respectively. We define $\mc{CN}(0,1)$ as
circularly symmetrical complex Gaussian distribution with zero mean
and unit variance.

\section{System Model}\label{sec-model}
Consider an $M\times N$ MIMO system with $M$ transmit antennas and
$N$ receive antennas. The channel coefficients from the transmitter
to the receiver are modeled as an independent and identically
distributed (i.~i.~d.) $\mc{CN}(0,1)$ Gaussian random variable. The
channels are assumed to be block fading, i.e., the fading
coefficients remain fixed for a constant number of channel-uses and
change independently from one block to another. We call the interval
under which channel coefficients remain unchanged the \emph{block
length} of the code. We assume that the block length is $L$
channel-uses.

We consider STBC designs with a maximum usage of $K$ independent
channel blocks. In other words, the system has a \emph{decoding
delay constraint} of $K$ blocks, or equivalently $T=KL$
channel-uses. The input-output relationship can be described by
\begin{align}\label{eq-sys}
\mb{Y}_k=\mb{H}_k\mb{X}_k+\mb{N}_k, k=0, \ldots, K-1,
\end{align}
where $\mb{Y}_k$, $\mb{H}_k$, $\mb{X}_k$, and $\mb{N}_k$ denote
$N\times L$ receive signal matrix, $N\times M$ channel matrix,
$M\times L$ transmit signal matrix, and $N\times L$ AWGN noise
matrix, respectively. Each entry of $\mb{N}_k$ is i.~i.~d.
$\mc{CN}(0,1)$ Gaussian distributed. The subscript $k$ denotes the
index of the block. Let
\begin{align}\label{eq-notation}
{\mb{Y}}=\left[\begin{array}{c}{\mb{Y}}_0\\ \vdots
\\{\mb{Y}}_{K-1}\end{array}\right], {\mb{X}}=\left[\begin{array}{c}{\mb{X}}_0\\ \vdots
\\{\mb{X}}_{K-1}\end{array}\right], {\mb{N}}=\left[\begin{array}{c}{\mb{N}}_0\\ \vdots
\\{\mb{N}}_{K-1}\end{array}\right],
\mb{H}=\left[\begin{array}{ccc}\mb{H}_0&\cdots&0\\
\vdots &\ddots &\vdots
\\0&\cdots&\mb{H}_{K-1}\end{array}\right].
\end{align}
The system equation in \eqref{eq-sys} can be combined as
\begin{align}\label{eq-sys3}
{\mb{Y}}=\mb{H}{\mb{X}}+{\mb{N}}.
\end{align}


The transmitter needs to send a set of $2^m$ messages
$\mc{M}=\{m_i\}$, where the subscript $i$ is used to represent the
message index. A set of $2^m$ codewords $\mc{X}_c$ is generated,
where each codeword $\mc{X}_i\in \mc{X}_c$ has dimension $KM\times
L$. For message $m_i$, the codeword $\mc{X}_i$ is selected to be
transmitted over the equivalent system in \eqref{eq-sys3}. Since
$2^m$ messages are sent in $KL$ channel-uses, the transmission rate
can be computed as $R=\frac{m}{KL}$ bits/channel-use. We assume that
the equivalent channel matrix $\mb{H}$ is noncausually known at the
transmitter when designing $\mc{X}_c$. Note that this CSIT
assumption is stronger than the causual CSIT assumption, i.e., the
transmitter can only know the past and current channels, but not the
future channels. The negative reuslts in this report is hence
applicable to the scenario of causual CSIT. With CSIT, each codeword
${\mc{X}}_i$ can be a function of $\mb{H}$, i.e. $\mc{X}_i(\mb{H})$.
For simplicity, we use $\mc{X}_i$ instead of $\mc{X}_i(\mb{H})$
throughout the rest of this report. Moreover, we have a \emph{block
power constraint} for each codeword $\mc{X}_i$. Mathematically, it
can be expressed as $\tr\{\mc{X}_i\mc{X}_i^*\}\le KLP$ for
$i=1,\ldots, 2^m$.

The receiver is assumed to have perfect channel information, and
decodes after receiving the entire block $\mb{Y}$. The
maximum-likelihood (ML) method is used to recover transmitted
message $i$ as,
\begin{align} \hat{i}=\arg\underset{i,{\mc{X}}_i\in
\mc{X}_c}{\max} \Prob ({\mb{Y}}|{\mb{X}=\mc{X}}_i,\mb{H}),
i=1,\ldots, 2^m.
\end{align}
A decoding error occurs when $\hat{i}\neq i$. Then, the probability
of average codeword decoding error can be defined as
\begin{align}\label{eq-wer}
\mr{P}_\mr{E}=\underset{\mb{H},\mb{X}}{\Exp}\Prob(\hat{i}\neq
i|\mb{X},\mb{H}),
\end{align}
where the expectation is taken over all channel realizations of
$\mb{H}$ and codewords in $\mc{X}_c$. In this report, we aim at
finding a lowerbound on $\mr{P}_\mr{E}$ for any STBC design with
CSIT.

\section{A Lowerbound on $\mr{P}_\mr{E}$}\label{sec-bound}

In this section, we explain the sphere-packing approach to obtain
the lowerbound on $\mr{P}_\mr{E}$. First, a geometrical
interpretation is introduced to provide some notations. Then, we
tackle the problem relying on these notations.

\begin{figure}[h]
\centering
  \includegraphics[width=5in]{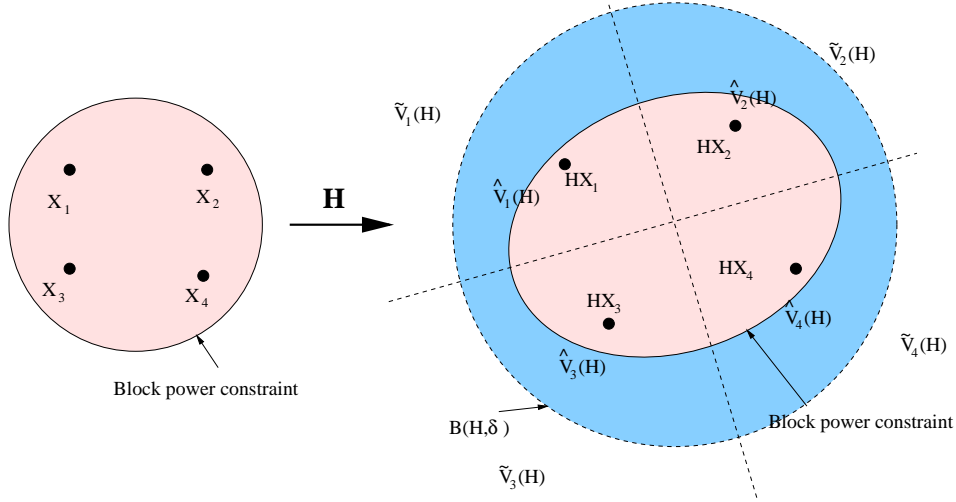}\\
  \caption{Geometric interpretation. The transmitter designs a set of four codewords with block power constraint (Left diagram). The receiver decodes by partitioning space into disjoint Voronoi Regions (Right diagram).}\label{Fig-partition}
\end{figure}

A geometrical interpretation of the system is described in
Fig.~\ref{Fig-partition}. Each transmit codeword $\mc{X}_i$ can be
interpreted as a point in $\mathbb{C}^{MLK}$. With block power
constraint, each codeword can only locate inside a hypersphere. The
equivalent channel $\mb{H}$ transforms $\mb{X}$ into a point in the
receive signal space $V$, which has dimension $\mathbb{C}^{NLK}$.
The transformation includes rotation and scale of the transmit
signal space. To illustrate the effect of transformation, the
hypersphere in the transmit signal space is transformed into a
hyper-ellipsoid in the receive signal space in
Fig.~\ref{Fig-partition}. The ML decoding method can be equivalently
described by partitioning the receive signal space $V$ into disjoint
\emph{Voronoi regions}. Let the Voronoi partition be
$V=\underset{i=1,\ldots,2^m}{\bigcup}V_i(\mb{H})$, where
$V_i(\mb{H})$ denotes the Voronoi region corresponding to
$\mc{X}_i$. A Voronoi region $V_i(\mb{H})$ is defined based on
distance metrics. For any point $\mb{Y}\in V_i(\mb{H})$, the
distance $\|{\mb{Y}}-\mb{H}{\mc{X}}_i\|$ is smaller than that of
$\|{\mb{Y}}-\mb{H}{\mc{X}}_j\|$ for all $j\neq i$. Then, an event of
decoding error can be equivalently interpreted as $\mb{Y}$ outside
the Voronoi region, i.e., $\mb{Y} \notin V_i(\mb{H})$ given
$\mc{X}_i$ is sent. Therefore, using the geometrical interpretation,
$\mr{P}_\mr{E}$ in \eqref{eq-wer} can be equivalently written as
\begin{align}\label{eq-wer2}
\mr{P}_\mr{E}=\underset{\mb{H}}{\Exp}\mr{P}_{\mr{E}|\mb{H}}=\underset{\mb{H}}{\Exp}\sum_i\Prob({\mb{Y}}\notin
V_i(\mb{H})|{\mb{X}=\mc{X}}_i,\mb{H})\Prob(\mb{X}=\mc{X}_i),
\end{align}
where $\mr{P}_{\mr{E}|\mb{H}}$ denotes decoding error probability
given the channel matrix $\mb{H}$.

We define some regions in the receive signal space $V$ to lowerbound
$\mr{P}_\mr{E}$. A hypersphere is defined as
\begin{align}\label{eq-sphere}
B(\mb{H},\delta)=\left\{\tr \left(\mb{Y}\mb{Y}^*\right)\le
\left(\sqrt{PLK \tr
\left(\mb{H}\mb{H}^*\right)}+\sqrt{NLK\delta}\right)^2 \right\},
\end{align}
where $\delta$ is any positive parameter to control the radius of
the hypersphere. The choice of $\delta$ will be discussed later. The
hypersphere defines a region that the receive signal resides with a
high probability due to block power constraint. Further, we denote
$\ol{B(\mb{H},\delta)}$ as the region outside the hypersphere
$B(\mb{H},\delta)$. Thus, the whole receive signal space can be
partitioned into the part inside hypersphere and the part outside
hypersphere, i.e., $V=B(\mb{H},\delta)\bigcup
\ol{B(\mb{H},\delta)}$. Since the volume of $V_i(\mb{H})$ may be
unbounded, we further partition each $V_i(\mb{H})$ into two parts
with reference to $B(\mb{H},\delta)$. We define
\begin{align}\label{eq-def}
\hat{V}_i(\mb{H})=V_i(\mb{H})\bigcap B(\mb{H},\delta), \quad
\tilde{V}_i(\mb{H})=V_i(\mb{H})\bigcap \ol{B(\mb{H},\delta)},
\end{align}
where $\hat{V}_i(\mb{H})$ is the part of $V_i(\mb{H})$ inside the
hypersphere $B(\mb{H},\delta)$ and $\tilde{V}_i(\mb{H})$ is the part
outside of the hypersphere. Fig.~\ref{Fig-partition} illustrates the
partition of receive signal space for a set of four codewords, i.e.,
$m=2$.

In what follows, we obtain a lowerbound on \eqref{eq-wer2}. For
simplicity, we can omit $\mb{H}$ in the notations used in
\eqref{eq-wer2}, \eqref{eq-sphere}, and \eqref{eq-def}. It follows
\begin{align}\nonumber
\mr{P}_{\mr{E}|\mb{H}}&=\sum_i\Prob({\mb{Y}}\notin
V_i|{\mb{X}}=\mc{X}_i)\Prob(\mb{X}=\mc{X}_i)\\ \nonumber &=\sum_i
\Prob(\mb{Y}\notin
\hat{V}_i|\mb{X}=\mc{X}_i)\Prob(\mb{X}=\mc{X}_i)-\sum_{i}
\Prob(\mb{Y}\in \tilde{V}_i|\mb{X}=\mc{X}_i)\Prob(\mb{X}=\mc{X}_i)\\
\nonumber &>\sum_i\Prob(\mb{Y}\notin
\hat{V}_i|\mb{X}=\mc{X}_i)\Prob(\mb{X}=\mc{X}_i)-\sum_i
\Prob(\mb{Y}\in
\ol{B(\delta)}|\mb{X}=\mc{X}_i)\Prob(\mb{X}=\mc{X}_i)\\
\label{eq-bound} &=\underset{A}{\underbrace{\sum_i\Prob(\mb{Y}\notin
\hat{V}_i|\mb{X}=\mc{X}_i)}}\Prob(\mb{X}=\mc{X}_i)- \Prob(\mb{Y}\in
\ol{B(\delta)}),
\end{align}
where we have Line $2$ since $\tilde{V}_i$ and $\hat{V}_i$ are
disjoint sets and $\tilde{V}_i\bigcup\hat{V}_i=V_i$; the inequality
in Line $3$ is true since $\tilde{V}_i$ is included in
$\ol{B(\delta)}$. The following two lemmas are needed to provide a
lowerbound on $\mr{P}_{\mr{E}}$.
\renewcommand{\baselinestretch}{1.4}
\begin{lemma}\label{lemma1}
Let $S(r_i)$ be an $(nLK)$-hypersphere centered at $\mb{H}\mc{X}_i$
with a radius of $r_i$. The radius $r_i$ is selected such that
$S(r_i)$ and $\hat{V}_i$ have the same volume. Substituting
$\hat{V}_i$ with $S(r_i)$ in $\Prob(\mb{Y}\notin
\hat{V}_i|\mb{X}=\mc{X}_i)$, we have
\begin{align*}
\Prob(\mb{Y}\notin \hat{V}_i|\mb{X}=\mc{X}_i)\ge \Prob(\mb{Y}\notin
S(r_i)|\mb{X}=\mc{X}_i)
\end{align*}
\end{lemma}
\begin{proof}
See \cite{Pos71}.
\end{proof}
Intuitively, this lemma can be explained using
Fig.~\ref{Fig-sphere}. The PDF of $\mb{Y}$ given $\mc{X}_i$ depends
only on the radius $r_i$. For any point in Region III, its PDF is
higher than that of any point in Region II. Then, the probability of
the receive signal in Region III is higher than that in Region II.
As a result, the probability of $\mb{Y}$ being inside the sphere
$S(r_i)$ is higher than that inside $\hat{V}_i$. Conversely, it is
less likely for $\mb{Y}$ to be outside $S(r_i)$ than $\hat{V}_i$.

\begin{figure}
\centering
  \includegraphics[width=3in]{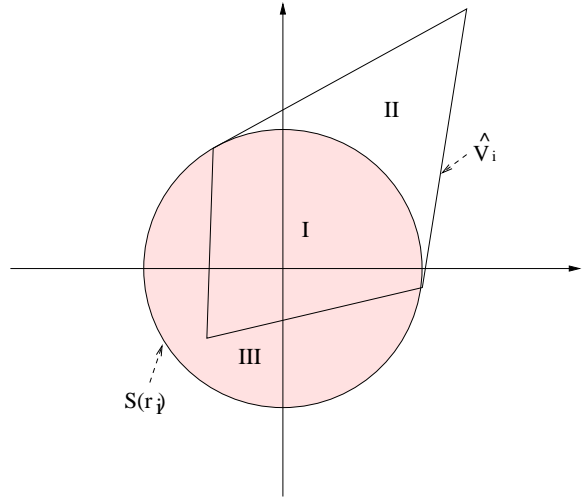}\\
  \caption{Lemma 1: Lowerbounds using sphere as Voronoi Region.}\label{Fig-sphere}
\end{figure}

\begin{lemma}\label{lemma2}
The probability of $\mb{Y}$ to be outside of the hypersphere
$B(\delta)$ is upperbounded by
\begin{align*}
\Prob\Big(\mb{Y}\in \ol{B(\delta)}\Big)\le
\frac{\Gamma(NLK,NLK\delta)}{\Gamma(NLK)},
\end{align*}
where $\Gamma(n,x)$ denotes the incomplete Gamma function, i.e.,
$\Gamma(n,x)=\int_{x}^{+\infty} t^{n-1}\mr{e} ^{-t} \mr{d}t$.
\end{lemma}
\renewcommand{\baselinestretch}{1.4}
\begin{proof}
From \eqref{eq-sys3}, $\mb{Y}$ is the sum of $\mb{H}\mb{X}$ and
$\mb{N}$. With a block power constraint, the Frobenius norm of the
first term can be bounded as $\tr (\mb{X}^*\mb{H}^*\mb{H}\mb{X})\le
\tr (\mb{X}^*\mb{X})\tr (\mb{H}^*\mb{H})=KLP\tr (\mb{H}^*\mb{H})$.
Then, we have
\begin{align*}
\Prob\Big(\tr (\mb{X}^*\mb{H}^*\mb{H}\mb{X})\le KLP\tr
(\mb{H}^*\mb{H})\Big)=1.
\end{align*}
From \eqref{eq-notation}, since $\mb{N}$ is the equivalent $NK\times
L$ noise matrix and each entry is i. i. d. $\mc{CN}(0,1)$
distributed, $\tr (\mb{N}\mb{N}^*)$ is Chi-square distributed with
$2NLK$ degrees of freedom. Then, we can compute
\begin{align*}
\Prob\Big(\tr (\mb{N}\mb{N}^*)\le NLK\delta\Big)=1-\frac{\Gamma(NLK,
NLK\delta)}{\Gamma(NLK)}.
\end{align*}

Since the norm of the sum of two matrices can be upperbounded by the
sum of the norms of each matrix, we have $\sqrt{\tr
(\mb{Y}^*\mb{Y})}\le \sqrt{\tr
(\mb{X}^*\mb{H}^*\mb{H}\mb{X})}+\sqrt{\tr (\mb{N}\mb{N}^*)}$. The
probability of $\mb{Y}$ falling into the hypersphere $B(\delta)$ can
be bounded as
\begin{align*}
\Prob(\mb{Y}\in B(\delta))=&\Prob\Big(\tr (\mb{Y}^*\mb{Y})\le
\left(\sqrt{KLP\tr
(\mb{H}^*\mb{H})}+\sqrt{NLK\delta}\right)^2\Big)\\
= & \Prob\Big(\sqrt{\tr (\mb{Y}^*\mb{Y})}\le \sqrt{KLP\tr
(\mb{H}^*\mb{H})}+\sqrt{NLK\delta}\Big)\\
\ge & \Prob\Big(\sqrt{\tr (\mb{X}^*\mb{H}^*\mb{H}\mb{X})}+\sqrt{\tr
(\mb{N}\mb{N}^*)}\le \sqrt{KLP\tr
(\mb{H}^*\mb{H})}+\sqrt{NLK\delta}\Big)\\
 \ge & \Prob\Big(\tr
(\mb{X}^*\mb{H}^*\mb{H}\mb{X}) \le KLP\tr
(\mb{H}^*\mb{H})\Big)\Prob\Big(\tr (\mb{N}\mb{N}^*) \le
NLK\delta\Big)=1-\frac{\Gamma(NLK, NLK\delta)}{\Gamma(NLK)}.
\end{align*}
Therefore, the probability that $\mb{Y}$ is outside the hypersphere
$B(\delta)$ is upperbounded by $\frac{\Gamma(NLK,
NLK\delta)}{\Gamma(NLK)}$.
\end{proof}
Lemma 2 says that with a block power constraint, the receive signal
$\mb{Y}$ is constrained in the hypersphere $B(\delta)$ with a high
probability. Proving Lemmas \ref{lemma1} and \ref{lemma2}, we are
ready for the following theorem.

\begin{theorem}\label{thm}
When the messages are equiprobable, i.e.,
$\Prob(\mb{X}=\mc{X}_i)=\frac{1}{2^m}$, a lowerbound on
$\mr{P}_{\mr{E}}$ can be obtained as\footnote{When messages are not
uniformly distributed, we can use $\min_i P(\mb{X}=\mc{X}_i)$ to
lowerbound $P(\mb{X}=\mc{X}_i)$. It is straightforward to extend the
results to the case of non-uniform messages. }
\begin{align}\mr{P}_{\mr{E}}\ge \frac{P^{NLK}\Gamma(NLK+MNK)}{\Gamma(NLK)\Gamma(MNK)}\int_{a}^{b}\frac{\mc{X}^{NLK-1}}{(1+P\mc{X})^{MNK+NLK}}
\mr{d}\mc{X},
\end{align}
where the bounds of the integral are
$a=LK2^{-R/N}\left(\frac{\sqrt{2}}{2^{R/(2N)}-1}+1\right)^2$,
$b=\frac{2LK}{\left(2^{R/(2N)}-1\right)^2}$.
\end{theorem}

\begin{proof}
When messages are equiprobable, from Lemma \ref{lemma1}, we can
further lowerbound the term $A$ in \eqref{eq-bound} as
\begin{align}\nonumber
A&=\frac{1}{2^m}\sum_{i}\Prob\Big(\mb{Y}\notin \hat{V}_i|\mb{X}=\mc{X}_i\Big)\\
&\ge\frac{1}{2^m}\sum_{i}\Prob\Big(\mb{Y}\notin
S(r_i)|\mb{X}=\mc{X}_i\Big)=\frac{1}{2^m}\sum_{i}\frac{\Gamma(NLK,
r_i^2)}{\Gamma(NLK)}.\label{eq-bound2}
\end{align}
Substituting \eqref{eq-bound2} and the result of Lemma \ref{lemma2}
into \eqref{eq-bound}, we can lowerbound $\mr{P}_{\mr{E}|\mb{H}}$ as
\begin{align}\label{eq-bound31}
\mr{P}_{\mr{E}|\mb{H}}\ge\frac{1}{2^m}\sum_{i}\frac{\Gamma(NLK,
r_i^2)}{\Gamma(NLK)}-\frac{\Gamma(NLK,NLK\delta)}{\Gamma(NLK)}.
\end{align}
In what follows, we further lowerbound the RHS of \eqref{eq-bound31}
to help analyze diversity. Note that $\hat{V}_i$ is inside the
hypersphere $B(\delta)$ and $\hat{V}_i$ has the same volume as
$S(r_i)$. There is an additional constraint on the sum of the
volumes of $S(r_i)$, i.e.,
$\mr{Vol}\Big(B(\delta)\Big)=\underset{i}{\sum}
\mr{Vol}\Big(S(r_i)\Big)$.  
Thus, an optimization problem can be formulated as
\begin{align}\nonumber
\min_{r_i}&\ \frac{1}{2^m}\sum_{i}\frac{\Gamma(NLK,
r_i^2)}{\Gamma(NLK)}-\frac{\Gamma(NLK,NLK\delta)}{\Gamma(NLK)}\\
\mr{s.t.}&\ \underset{i}{\sum}\nonumber
\mr{Vol}\Big(S(r_i)\Big)=\mr{Vol}\Big(B(\delta)\Big),\\
&\delta >0.\label{eq-opt}
\end{align}
Since $\delta$ is a constant, the second term in the objective
function can be ignored. This minimization problem is known as
\emph{sphere-packing}. The solution is obtained when each
hypersphere $S(r_i)$ has equal radius, i.e., $r_i=r_\mr{o}(\delta)$
for $i=1,\ldots, 2^m$\cite{Fo03}. Since the volume of a hypersphere
in $\mathbb{C}^{n}$ is $\mr{Vol}=R_nr^{2n}$ where $R_n$ is a
constant depending on dimensions and $r$ is the radius, from the
constraint, we have
\begin{align*}
2^m R_n r_\mr{o}(\delta)^{2NLK}=R_n \left(\sqrt{PLK\tr
(\mb{H}\mb{H}^*)}+\sqrt{NLK\delta}\right)^{2NLK}.
\end{align*}
After dividing both sides of the above equation by $2^mR_n$ and
taking $NLK$th root, we have
\begin{align}\label{eq-radius}
r_\mr{o}(\delta)^2=NLK2^{-\frac{R}{N}}
\left(\sqrt{\delta}+\sqrt{\frac{P}{N}\tr (\mb{H}\mb{H}^*)}\right)^2,
\end{align}
where $R=\frac{m}{LK}$, denoting the bit-rate per channel-uses of
the STBC designs. Replacing \eqref{eq-radius} into
\eqref{eq-bound31}, we have a lowerbound on $\mr{P}_{\mr{E}|\mb{H}}$
as
\begin{align}
\mr{P}_{\mr{E}|\mb{H}}\ge \frac{\Gamma(NLK,
r_\mr{o}(\delta)^2)}{\Gamma(NLK)}-\frac{\Gamma(NLK,NLK\delta)}{\Gamma(NLK)}.
\label{eq-bound3}
\end{align}

In what follows, we discuss the choice of $\delta$. In \cite{Fo03},
$\delta$ is chosen to be one, and the probability of $\mb{Y}$ to be
outside the hypersphere, i.e., the second term in Line $4$ of
\eqref{eq-bound}, is not taken into account. Therefore, the
lowerbound obtained in \cite{Fo03} is not a tight lowerbound. In
this report, we formulate the lowerbound considering both the events
when $\mb{Y}$ is inside and outside the hypersphere. 
From \eqref{eq-radius}, the radius of each hypershpere $r_\mr{o}$
depends on $\delta$, and for any positive $\delta$, the RHS of
\eqref{eq-bound3} provides a new lowerbound. To find an explicit
lowerbound on $\mr{P}_{\mr{E}}$, we choose a $\delta$ that results
in a positive number on the RHS of \eqref{eq-bound3}\footnote{Note
that the best lowerbound can be found by futher maximizing
\eqref{eq-bound3} with respect to $\delta$. The optimal $\delta$ can
be obtained by calculating the derivative of \eqref{eq-bound3} and
set the derivative to zero. The resulting equality is nonlinear in
$\delta$, and cannot be solved explicitely. Therefore, we cannot
find an explicit lowerbound on $\mr{P}_{\mr{E}}$ using the optimal
$\delta$.}. Note that the first and second terms in
\eqref{eq-bound3} are both incomplete Gamma functions with $NLK$
{{degrees of freedom}}. To have a positive lowerbound, we need
\begin{align}\label{eq-req}
r_\mr{o}(\delta)^2<NLK\delta.
\end{align}
Substituting \eqref{eq-radius} into \eqref{eq-req}, we have
\begin{align*}
&2^{-\frac{R}{N}} \left(\sqrt{\delta}+\sqrt{\frac{P}{N}\tr
(\mb{H}\mb{H}^*)}\right)^2<\delta\\
&\frac{\delta}{\frac{P}{N}\tr
(\mb{H}\mb{H}^*)}>\frac{1}{\left(2^{\frac{R}{2N}}-1\right)^2}.
\end{align*}
Then, we let $\delta=\frac{2P}{N(2^{\frac{R}{2N}}-1)^2}\tr
(\mb{H}\mb{H}^*)$. We can expand the RHS of \eqref{eq-bound3} into
an integral as
\begin{align}\nonumber
\mr{P}_{\mr{E}|\mb{H}}&\ge\frac{1}{\Gamma(NLK)}\int_{r_\mr{o}^2}^{NLK\delta}x^{NLK-1}\mr{e}^{-x}
\mr{d}x\\
&=\frac{1}{\Gamma(NLK)}\int_{PLK2^{-R/N}\left(\frac{\sqrt{2}}{2^{R/(2N)}-1}+1\right)^2\tr
(\mb{H}\mb{H}^*)}^{\frac{2PLK}{\left(2^{R/(2N)}-1\right)^2}\tr
(\mb{H}\mb{H}^*)}x^{NLK-1}\mr{e}^{-x} \mr{d}x.\label{eq-bound4}
\end{align}
Further let
$a=LK2^{-R/N}\left(\frac{\sqrt{2}}{2^{R/(2N)}-1}+1\right)^2$,
$b=\frac{2LK}{\left(2^{R/(2N)}-1\right)^2}$, and $h=\tr
(\mb{H}\mb{H}^*)$ to simplify the notations, and integrate
\eqref{eq-bound4} over $h$. From \eqref{eq-notation}, since $\tr
(\mb{H}\mb{H}^*)$ is Chi-square distributed with $2MNK$ degrees of
freedoms, we obtain a lowerbound on $\mr{P}_{\mr{E}}$ as
\begin{align}\nonumber
\mr{P}_{\mr{E}}&\ge\frac{1}{\Gamma(NLK)\Gamma(MNK)}\int_{0}^{+\infty}h^{MNK-1}\mr{e}^{-h}\int_{aPh}^{bPh}x^{NLK-1}\mr{e}^{-x}
\mr{d}x\mr{d}h\\ \nonumber
&=\frac{1}{\Gamma(NLK)\Gamma(MNK)}\int_{0}^{+\infty}h^{MNK+NLK-1}P^{NLK}\mr{e}^{-h}\int_{a}^{b}\mc{X}^{NLK-1}\mr{e}^{-Ph\mc{X}}
\mr{d}\mc{X}\mr{d}h\\ \nonumber
&=\frac{P^{NLK}}{\Gamma(NLK)\Gamma(MNK)}\int_{a}^{b}\mc{X}^{NLK-1}
\mr{d}\mc{X}\int_{0}^{+\infty}
h^{MNK+NLK-1}\mr{e}^{-h(1+P\mc{X})}\mr{d}h\\
\label{eq-int}&=\frac{P^{NLK}\Gamma(NLK+MNK)}{\Gamma(NLK)\Gamma(MNK)}\int_{a}^{b}\frac{\mc{X}^{NLK-1}}{(1+P\mc{X})^{MNK+NLK}}
\mr{d}\mc{X}
\end{align}
In Line $2$, we have replaced the variable $x$ with $\mc{X}Ph$. This
concludes the proof.
\end{proof}

The proof of Theorem \ref{thm} uses the sphere-packing argument. The
diversity performance is given in the following corollary.

\begin{corollary}\label{cor}
With short-term decoding delay, block power, and block rate
constraints, the diversity of any STBC design is upperbounded by
$MNK$. It is independent of whether CSI is available at the
transmitter or not.
\end{corollary}
\begin{proof}
The diversity gain is defined as
\begin{align}\label{eq-div}
d=-\lim_{P\rightarrow \infty}\frac{\log \mr{P}_{\mr{E}}}{\log P}.
\end{align}
Replacing the lowerbound obtained in Theorem \ref{thm} into
\eqref{eq-div} results in an upperbound on diversity gain. Thus, we
have
\begin{align*}
d & \le -\lim_{P\rightarrow \infty}\frac{\log
\left(\frac{P^{NLK}\Gamma(NLK+MNK)}{\Gamma(NLK)\Gamma(MNK)}\int_{a}^{b}\frac{\mc{X}^{NLK-1}}{(1+P\mc{X})^{MNK+NLK}}
\mr{d}\mc{X}\right)}{\log P}
\\
&=-NLK-\lim_{P\rightarrow \infty}\frac{\log
\left(\int_{a}^{b}\frac{\mc{X}^{NLK-1}}{(1+P\mc{X})^{MNK+NLK}}
\mr{d}\mc{X}\right)}{\log P}\\
&\le-NLK-\lim_{P\rightarrow \infty}\frac{\log
\left(\frac{1}{(1+Pb)^{MNK+NLK}}\int_{a}^{b}\mc{X}^{NLK-1} \mr{d}\mc{X}\right)}{\log P}\\
&=-NLK-\lim_{P\rightarrow \infty}\frac{\log
\left(\frac{1}{(1+Pb)^{MNK+NLK}}\right)}{\log P}=MNK.
\end{align*}
In Line $3$, we lowerbound the integrand
$\frac{\mc{X}^{NLK-1}}{(1+P\mc{X})^{MNK+NLK}}$ by
$\frac{\mc{X}^{NLK-1}}{(1+Pb)^{MNK+NLK}}$; In Line $4$, we ignore
the integral because for a fixed-rate $R$, the bounds $a$ and $b$
are independent of power $P$.

Note that in all previous lowerbounding steps, the inequalities are
independent of CSIT. Therefore, the derived lowerbound is
independent of CSIT and is applicable to both CSIT and no CSIT
scenarios.
\end{proof}

The lowerbounding techniques used in the proof of Corollary
\ref{cor} motivate two explicit lowerbounds where no integral is
involved. From \eqref{eq-int}, we have the following two bounds.
\begin{align}\nonumber
Bound\ 1: &
\mr{P}_{\mr{E}}>\frac{\Gamma(NLK+MNK)}{\Gamma(NLK)\Gamma(MNK)}\frac{P^{NLK}}{(1+Pb)^{MNK+NLK}}\int_{a}^{b}\mc{X}^{NLK-1}
\mr{d}\mc{X}\\ \label{eq-lowerbound1}
&=\frac{\Gamma(NLK+MNK)}{\Gamma(NLK+1)\Gamma(MNK)}\frac{P^{NLK}}{(1+Pb)^{MNK+NLK}}\left(b^{NLK}-a^{MLK}\right).\\
\nonumber Bound\ 2: &
\mr{P}_{\mr{E}}>\frac{\Gamma(NLK+MNK)}{\Gamma(NLK)\Gamma(MNK)}{a}^{NLK-1}P^{NLK}\int_{a}^{b}
\frac{1}{(1+P\mc{X})^{MNK+NLK}}\mr{d}\mc{X}\\ \label{eq-lowerbound2}
&=\frac{\Gamma(NLK+MNK-1)}{\Gamma(NLK)\Gamma(MNK)}a^{NLK-1}P^{NLK}\left((1+Pa)^{-NLK-MNK+1}-(1+Pb)^{-NLK-MLK+1}\right).
\end{align}

\section{Numerical results}\label{sec-numeric}
In this section, we demonstrate the results of our analysis and
compare it with the results of beamforming schemes using
singular-value decomposition (SVD). We consider two systems: System
A with parameters $M=2,N=1$; System B with parameters $M=2,N=2$;
$L=1,K=1$ for both Systems A and B. For these two simple cases, we
assume no channel coding, and the decoding error probability in
\eqref{eq-wer} is equivalent to the symbol error rate (SER). We
apply BPSK, QPSK, 8QAM modulation for the SVD schemes, whose
corresponding rates are $R=1, 2, 3$ bits per channel-uses,
respectively. We plot the two strict lowerbounds in
\eqref{eq-lowerbound1} and \eqref{eq-lowerbound2}.
\begin{figure}
\centering
  \includegraphics[width=3in,angle=-90]{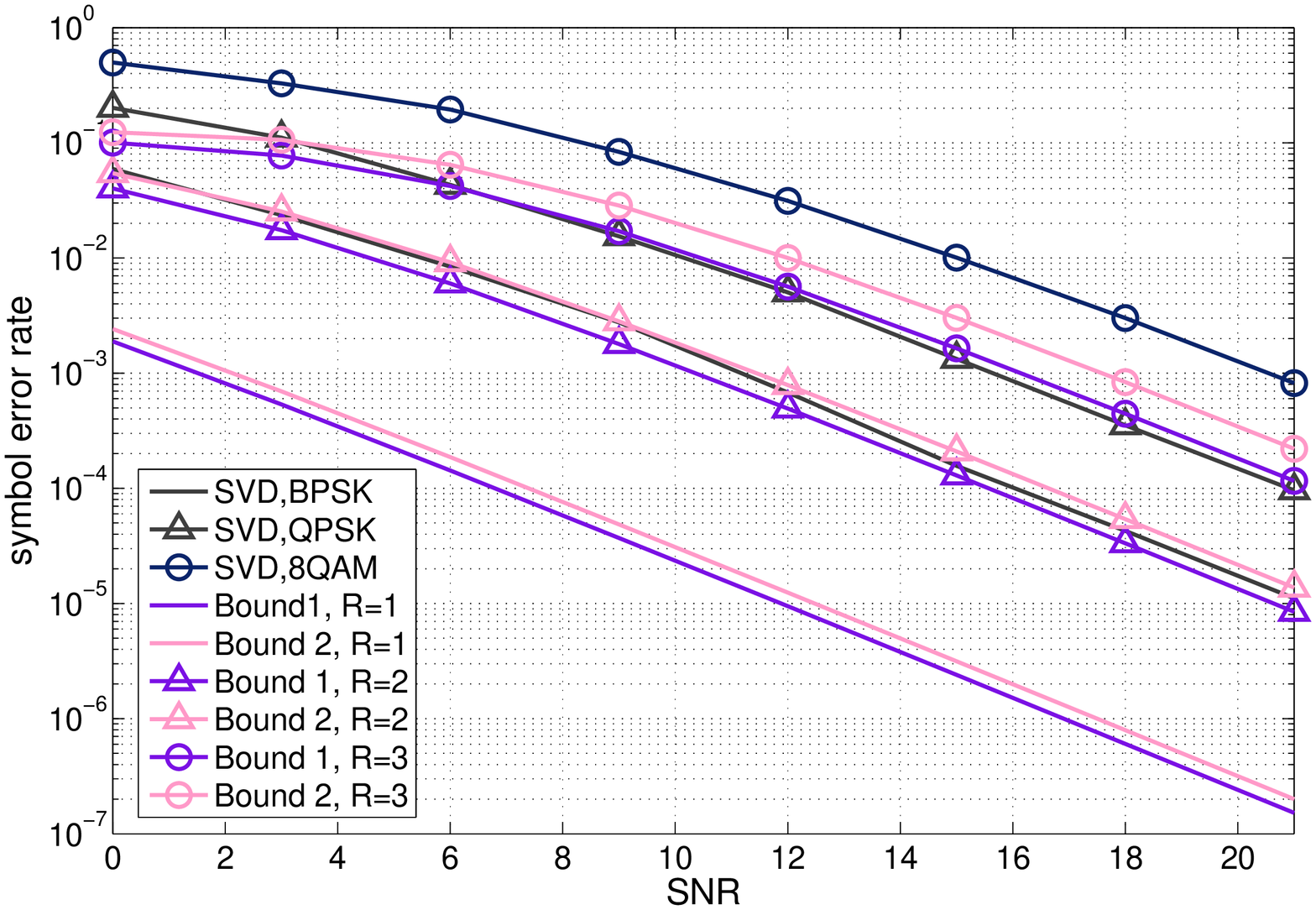}\\
  \caption{SER in a $2\times 1$ system with one channel-uses.}\label{fig-M2N1}
  \includegraphics[width=3in,angle=-90]{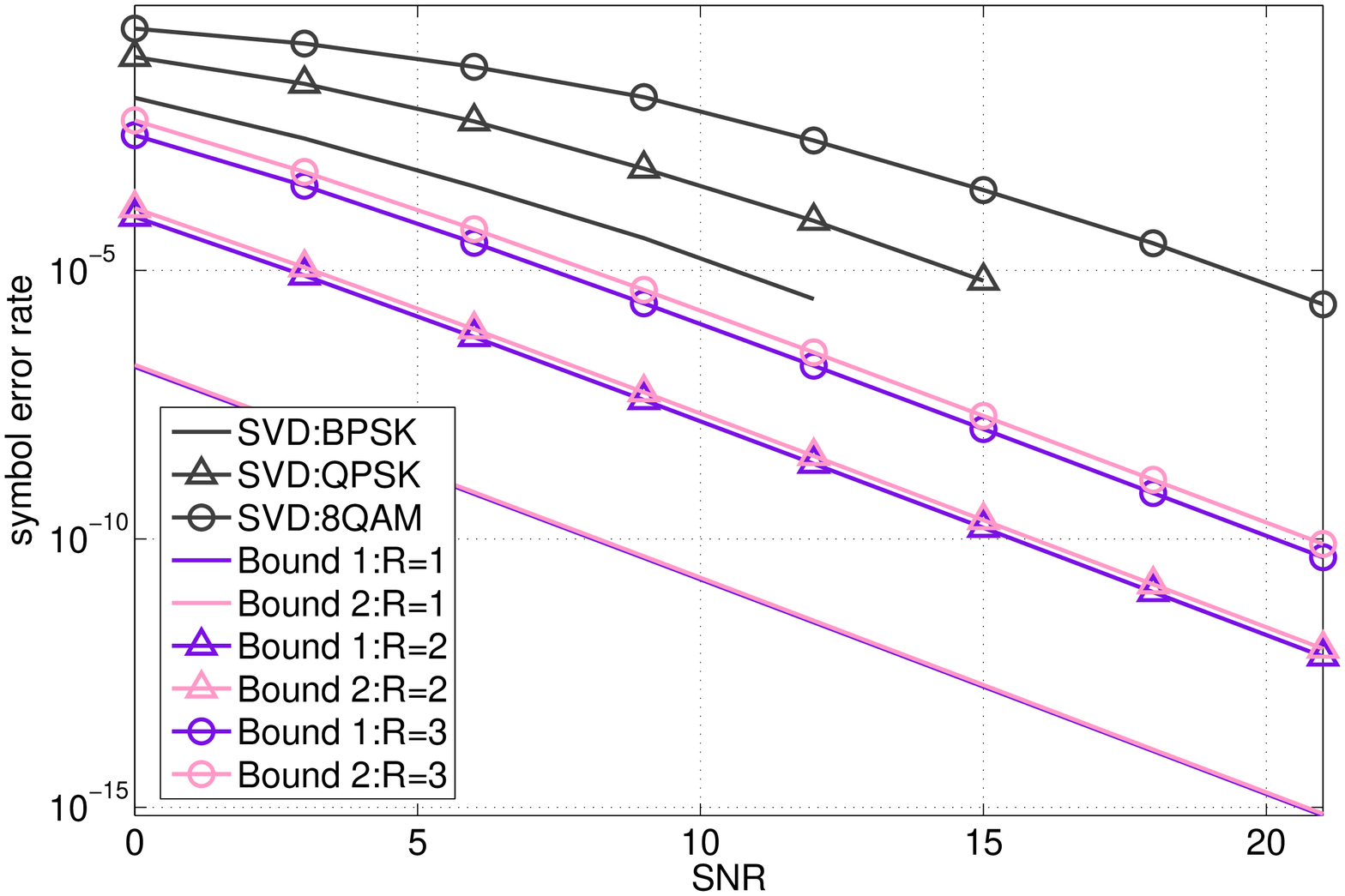}\\
  \caption{SER in a $2\times 2$ system with one channel-uses.}\label{fig-M2N2}
\end{figure}

Figs. \ref{fig-M2N1} and \ref{fig-M2N2} compare the SER of SVD
schemes with derived lowerbounds in Systems A and B, respectively.
Although the lowerbounds are loose in terms of array gain, the
diversity gain can be observed to be tight.

\section{Conclusion}\label{sec-conclusion}
We have obtained a negative result for STBC systems using
transmitter control with short-term constraints on decoding delay,
block power, and block rates. The analysis shows that fading cannot
be completely combatted with short-term constraints. The diversity
is upperbounded by the product of the numbers of transmit antennas,
receive antennas, and the independent fading block channels that
messages span over.

\footnotesize
\bibliographystyle{ieeetran}
\bibliography{IEEEabrv,STBC}

\begin{thebibliography}{1}
\providecommand{\url}[1]{#1}
\csname url@samestyle\endcsname
\providecommand{\newblock}{\relax}
\providecommand{\bibinfo}[2]{#2}
\providecommand{\BIBentrySTDinterwordspacing}{\spaceskip=0pt\relax}
\providecommand{\BIBentryALTinterwordstretchfactor}{4}
\providecommand{\BIBentryALTinterwordspacing}{\spaceskip=\fontdimen2\font plus
\BIBentryALTinterwordstretchfactor\fontdimen3\font minus
  \fontdimen4\font\relax}
\providecommand{\BIBforeignlanguage}[2]{{%
\expandafter\ifx\csname l@#1\endcsname\relax
\typeout{** WARNING: IEEEtran.bst: No hyphenation pattern has been}%
\typeout{** loaded for the language `#1'. Using the pattern for}%
\typeout{** the default language instead.}%
\else
\language=\csname l@#1\endcsname
\fi
#2}}
\providecommand{\BIBdecl}{\relax}
\BIBdecl

\bibitem{Ca72}
J.~Cavers, ``Variable-rate transmission for rayleigh fading channels,''
  \emph{IEEE Transactions on Communications}, vol.~20, no.~1, pp. 15 -- 22,
  Feb. 1972.

\bibitem{Cai99}
G.~Caire, G.~Taricco, and E.~Biglieri, ``Optimum power control over fading
  channels,'' \emph{IEEE Transactions on Information Theory}, vol.~45, no.~5,
  pp. 1468 --1489, Jul. 1999.

\bibitem{Sha08}
V.~Sharma, K.~Premkumar, and R.~Swamy, ``Exponential diversity achieving
  spatio-temporal power allocation scheme for fading channels,'' \emph{IEEE
  Transactions on Information Theory}, vol.~54, no.~1, pp. 188 --208, Jan.
  2008.

\bibitem{Gallager}
R.~Gallager, \emph{Information Theory and Reliable Communication}.\hskip 1em
  plus 0.5em minus 0.4em\relax New York: Wiley, 1968.

\bibitem{Ko05}
T.~Kotchiev, T.~Guess, and M.~McCloud, ``On finite-codelength outage capacity
  for {MIMO} channels,'' in \emph{2005 International Conference on Wireless
  Networks, Communications and Mobile Computing}, vol.~2, Jun. 2005, pp. 1047
  -- 1052 vol.2.

\bibitem{Wozencraft}
J.~Wozencraft and I.~Jacobs, \emph{Principles of Communication
  Engineering}.\hskip 1em plus 0.5em minus 0.4em\relax New York: Wiley, 1965.

\bibitem{Fo03}
M.~Fozunbal, S.~McLaughlin, and R.~Schafer, ``On performance limits of
  space-time codes: a sphere-packing bound approach,'' \emph{IEEE Transactions
  on Information Theory}, vol.~49, no.~10, pp. 2681 -- 2687, Oct. 2003.

\bibitem{Pos71}
G.~Poscetti, ``An upper bound for probability of error related to a given
  decision region in {N}-dimensional signal set (corresp.),'' \emph{IEEE
  Transactions on Information Theory}, vol.~17, no.~2, pp. 203 -- 206, Mar.
  1971.

\end{thebibliography}
\end{document}